\documentclass{ieeeaccess}

\usepackage[utf8]{inputenc}
\usepackage[cmex10]{amsmath}
\usepackage{graphicx}
\usepackage{amssymb}
\usepackage{amsthm}
\usepackage{subfigure}
\usepackage{cite}
\usepackage{color}
\usepackage{caption}
\usepackage{setspace}
\usepackage{subfigure}
\usepackage{amsmath}
\usepackage{tabu}
\usepackage{array}
\usepackage{booktabs}
\usepackage{amssymb}
\usepackage{threeparttable}
\usepackage[linesnumbered,ruled]{algorithm2e}
\usepackage{float}
\usepackage{stfloats}
\usepackage{lipsum}
\usepackage[english]{babel}
\usepackage{epstopdf}
\newtheorem{defn}{Definition}

\newtheorem{thm}{Lemma}
\newcommand{\argmax}{\operatornamewithlimits{argmax}}

\def\BibTeX{{\rm B\kern-.05em{\sc i\kern-.025em b}\kern-.08em
    T\kern-.1667em\lower.7ex\hbox{E}\kern-.125emX}}
\begin{document}

\doi{DOI}

\title{Fast Channel Estimation and Beam Tracking for Millimeter Wave Vehicular Communications}

\author{\uppercase{sina shaham}\authorrefmark{1}, \uppercase{Ming Ding}\authorrefmark{2}, \uppercase{Matthew Kokshoorn}\authorrefmark{3}, \uppercase{Zihuai Lin}\authorrefmark{4},
\uppercase{Shuping Dang\authorrefmark{5}, and Rana Abbas}.\authorrefmark{6}}

\address[1]{Department of Engineering, The University of Sydney, Sydney, NSW, 2006 Australia}
\address[2]{Data61, Sydney, NSW, 1435 Australia}
\address[3]{Department of Engineering, The University of Sydney, Sydney, NSW, 2006 Australia}
\address[4]{Department of Engineering, The University of Sydney, Sydney, NSW, 2006 Australia}
\address[5]{Computer, Electrical and Mathematical Sciences and Engineering Division, King Abdullah University of Science and Technology, Thuwal 23955-6900, Saudi Arabia}
\address[6]{Department of Engineering, The University of Sydney, Sydney, NSW, 2006 Australia}

\markboth
{S. Shaham \headeretal: Fast Channel Estimation and Beam Tracking for Millimeter Wave Vehicular Communications}
{S. Shaham \headeretal: Fast Channel Estimation and Beam Tracking for Millimeter Wave Vehicular Communications}

\begin{abstract}
Millimeter wave (mmWave) has been claimed to be the only viable solution for high-bandwidth vehicular communications. However, frequent channel estimation and beamforming required to provide a satisfactory quality of service limits mmWave for vehicular communications. In this paper, we propose a novel channel estimation and beam tracking framework for mmWave communications in a vehicular network setting. For channel estimation, we propose an algorithm termed robust adaptive multi-feedback (RAF) that achieves comparable estimation performance as existing channel estimation algorithms, with a significantly smaller number of feedback bits. We derive upper and lower bounds on the probability of estimation error (PEE) of the RAF algorithm, given a number of channel estimations, whose accuracy is verified through Monte Carlo simulations. For beam tracking, we propose a new practical model for mmWave vehicular communications. In contrast to the prior works, the model is based on position, velocity, and channel coefficient, which allows a significant improvement of the tracking performance. Focused on the new beam tracking model, we re-derive the equations for Jacobian matrices, reducing the complexity for vehicular communications. An extensive number of simulations is conducted to show the superiority of our proposed channel estimation method and beam tracking algorithm in comparison with the existing algorithms and models. Our simulations suggest that the RAF algorithm can achieve the desired PEE, while on average, reducing the feedback overhead by $75.5\%$ and the total channel estimation time by $14\%$. The beam tracking algorithm is also shown to significantly improve beam tracking performance, allowing more room for data transmission.
\end{abstract}
\titlepgskip=-15pt

\maketitle

\section{Introduction}\label{Introduction}

\PARstart{M}{illimeter wave}\footnote{This work was submitted in part and accepted to appear in the proceedings of IEEE Wireless Communications and Networking Conference, 2018~\cite{shaham2018raf}.}
 wireless communications is one of the primary candidates proposed to cater for the high data traffic demand of 5G mobile network \cite{ref1}, \cite{ref2}. The mmWave spectrum is considered to be from 30 to 300 GHz, which enables high-rate data transmission. Current research in mmWave is mostly focused on the 28 GHz, 38 GHz, and 60 GHz bands as well as the E-band, consisting of 71-76 GHz and 81-86 GHz \cite{sss1}. Several standards have already been established to regulate the use of mmWave, such as ECMA-387\cite{s2}, IEEE 802.15.3.c\cite{s3}, and more importantly, IEEE 802.11ad\cite{s4}, which is the first standard in the IEEE 802.11 family to support a mmWave band, i.e., 60 GHz band.

Exploiting the high data rate of mmWave paves the way for a number of exciting applications, such as mmWave cellular systems, vehicle to vehicle (V2V) communications, and vehicle to infrastructure (V2I) communications. Conventional protocols for vehicular communications fail at providing the high data rate required for many of its applications, e.g., high-resolution map downloads for navigation, collection and distribution of aggregated sensor information from/to vehicles for improved safety, clouding computing of the transmitted data from vehicles, etc. For instance, dedicated short range communications (DSRC) provides only 2-6 Mb/s for a range of 1000 m, and cellular communications offer at most 100 Mb/s in high-mobility scenarios. On the other hand, the existence of line of sight (LoS) paths in vehicular communications with high probability makes this high-bandwidth technology more suitable. That is because the height of a base station (BS) is usually much higher than that of vehicles with embedded transceivers mounted on top. Moreover, the limited communication range provides an inherent security feature.

It is worth noting that the use of mmWave for vehicular communications is not a new concept \cite{r4}. However, it is only in recent years that the advancements in CMOS technologies used in radio frequency integrated circuits have made the concept practical. Nonetheless, the applications of mmWave for vehicular communications still faces a number of open challenges, e.g., high path loss, limited communication range, beam training and alignment overhead due to the high mobility, etc. In this paper, we focus on the latter challenge. To compensate for the high path loss, mmWave is heavily dependent on establishing directional links with high beamforming gains. This requires frequent and accurate channel estimation and tracking reports. Moreover, having large antenna arrays increases the complexity of channel estimation as well as the number of required feedback bits.  On the other hand, the high mobility feature of vehicles leads to a fast changing environment which increases the frequency of channel estimation even further. Therefore, significantly faster, more reliable and more robust techniques are required to allow for sufficiently reliable and efficient data transfer between transmitters and receivers in vehicular communications, compared to conventional and stationary applications.

In the rest of this section, we review the current state of art approaches, and then, we briefly present our contributions in this paper.

\subsection{Related Work}

In general, channel estimation in mmWave is focused on finding three parameters: the angle of arrival (AoA), the angle of departure (AoD) and the channel coefficient ($\alpha$). Recent measurements have demonstrated a sparse nature of mmWave communication channels\cite{ref6}. Exploiting the sparsity, several works such as \cite{c2} and \cite{RH} have shown the efficiency of compressive sensing methods in decreasing the training overhead required for channel estimation. Authors in \cite{c1} proposed a hierarchical multi-resolution beamforming codebook to estimate the channel. In \cite{RH}, the authors developed a multi-stage adaptive channel estimation algorithm. In each stage, the possible AoA and AoD are divided into two subspaces ($K = 2$), and the most likely subspaces are chosen for further refinement in the next stage. The channel coefficient is estimated after the best link has been found. However, compressive sensing is well known to be non-adaptive technology. In \cite{overlap} authors followed the same approach with $K = 3$ to improve performance while maintaining low complexity and speed through using overlapped beam-patterns. One major challenge is that if in any of the stages the estimated angle is incorrect, the estimation in the following stages will also be incorrect due to the error propagation effect. However, if we have an insight into the probability of estimation error (PEE) associated with each measurement, we can terminate channel estimation when the PEE is below a predetermined threshold, as with rate adaptive algorithm RACE developed in \cite{RACE}. Unfortunately, RACE requires a large number of feedback bits, particularly in the low signal-to-noise ratio (SNR) regime.

After channel estimation, to prolong the duration of communication between the transmitter and receiver, fast beam tracking methods are required. This is practical for mmWave vehicular communications where vehicles are likely to move at uniform speeds for sufficiently long periods of time. Authors in \cite{midc} proposed a beamforming protocol for 60-GHz propagation channels. The method exploited training sequences to detect signal strengths. The evaluation of the proposed algorithm was provided in \cite{midc2}. The approach required multiple beam training sequences. In \cite{zhang}, the focus of the paper is on tracking the beams obtained by a full scan of all possible beam directions. The proposed algorithm applies the EKF to track paths. This method required a high overhead of pilot transmission to attain the measurement matrix. Moreover, the state model is based on angles only, without considering given to the channel coefficient. In \cite{rr}, the authors improved the tracking by having a single measurement instead of the full scan, which reduces the overall overhead. However, as in \cite{zhang}, the change in angles was modeled as a Gaussian noise with zero mean, which we will show later is not valid for vehicular communications.

Garcia et al.~\cite{garcia2016location} investigated the challenges associated with the mmWave communications in the vehicular domain. The authors proposed a location-aided beamforming strategy and analyzed the resulting performance in terms of antenna gain and latency. The outcome of experiments indicated the significance of location information for the channel estimation and beam tracking. Vutha et al.~\cite{va2017inverse} considered beam alignment in mmWave communications of vehicular settings. This paper proposed to use the vehicle’s position (e.g., available via GPS) to query a multipath
fingerprint database, which provides prior knowledge of potential pointing directions for reliable beam alignment. The approach is the inverse of fingerprinting localization, where the measured multipath signature is compared to the fingerprint database to retrieve the most likely position. Xinyu et al.~\cite{gao2016fast} focused on beam selection in terahertz (THz) massive MIMO systems. The authors proposed to utilize the obtained beamspace channels in the previous time slots to predict the prior
information of the beamspace channel in the following time slot without channel estimation.

\subsection{Main Contributions}
This paper aims to propose a novel approach for channel estimation and beam tracking in V2I mmWave communication systems to maximize the communication time of receiver (RX) and transmitter (TX). Our main contributions are summarized as follows.

\begin{itemize}
  \item We propose a new multi-stage adaptive algorithm referred to as robust adaptive multi-feedback (RAF) for mmWave channel estimation. The main advantage of the proposed algorithm is its low feedback overhead. Then, we derive the closed-form expression for the minimum number of feedback bits required for channel estimation, under maximum likelihood decoding, for a given probability of estimation error (PEE) constraint. We show that the estimation performance yielded by the RAF algorithm performs close to that bound, at a significantly reduced complexity.
  \item We derive upper and lower bounds on the PEE for the proposed RAF algorithm. The accuracy of these bounds is verified via Monte Carlo simulations.
  \item We show that the existing model used for beam tracking in mmWave using EKF recursion is not suitable for vehicular communications. Accordingly, we propose new evolution and observation models for beam tracking using EKF recursion, with the derivation of closed-form expressions for Jacobian matrices. The derived expressions are shown to be of less complexity, especially in the calculation of Jacobian matrices.\\
\end{itemize}

\textit{Notation} : Capital bold-face letter ($\boldsymbol{A}$) is used to denote a matrix, $\boldsymbol{a}$ to denote a vector, ${a}$ to denote a scalar and $\mathcal{A}$ denotes a set. $||\boldsymbol{A}||^2$ is the magnitude of $\boldsymbol{A}$, $|a|$ is the absolute value of $a$, and determinant is shown by $\text{det}(\boldsymbol{A})$. $\boldsymbol{A}^T$, $\boldsymbol{A}^H$ and $\boldsymbol{A}^*$ are the transpose, conjugate transpose and conjugate of $\boldsymbol{A}$, respectively. For a square matrix $\boldsymbol{A}$, $\boldsymbol{A}^{-1}$ represents its inverse matrix. $\boldsymbol{I}_N$ is the $N\times N$ identity matrix and $\lceil \cdot \rceil$ denotes the ceiling function. The superscripts $(.)^\textrm{R}$, $(.)^\textrm{I}$ return real and imaginary parts of the complex number enclosed, respectively. $\mathcal{C}\mathcal{N}(\boldsymbol{m},\boldsymbol{R})$ is a complex Gaussian random vector with mean $\boldsymbol{m}$ and covariance matrix $\boldsymbol{R}$, and $\text{E}[\boldsymbol{a}]$ and $\text{Cov}[\boldsymbol{a}]$ denote the expected value and covariance of ${\boldsymbol{a}}$, respectively.

\section{System model}\label{System model}
\subsection{mmWave V2I Communication System}
\subsubsection{Overview}
We consider a BS that is installed on a cellular tower or building with a height of $h$, as depicted in Fig. \ref{f1}. At any transmission block $j$, the position of the vehicle is represented by $d_j$ (point $B$), where $d_j$ is equal to the distance between the vehicle and the perpendicular line connecting the antenna array to the ground. Moreover, the speed of the vehicle is denoted by $v_j$, and its RX angle is denoted by $\theta_j$. In the following transmission block ($j+1$), the vehicle would have, thus, moved to position $d_{j+1}$ (point $C$). Similarly, its speed and receiving angle are now $v_{j+1}$ and $\theta_{j+1}$, respectively. The corresponding TX angles at BS are denoted by $\phi_j$ and $\phi_{j+1}$. The TX and RX angles are chosen to be the angles between the positive $x$-axis and the line connecting the receiver to the transmitter ($BA$). Furthermore, the receive antenna array is mounted on the roof-top of the vehicle which results in a dominant LoS path between the transceivers.

\begin{figure}[t!]
\centering
\includegraphics[scale=.58]{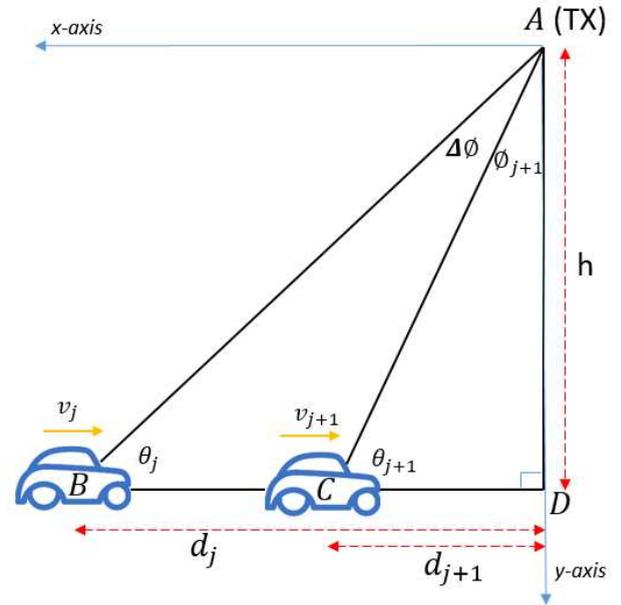}
\caption{Schematic Diagram of V2I model considered in this paper.}
\label{f1}
\end{figure}

\subsubsection{Transmission Scheme}

The communication protocol between the vehicle and the BS is illustrated in Fig. \ref{TB}. A beacon interval is defined as the maximum time period before a new channel estimation is required. We assume that the beacon interval is made up of $(m+1)$ discrete transmission blocks ($m=0,...,M$), with an equal duration $\Delta t$. Initial channel estimation takes place at the beginning of the first block ($m=0$), followed by channel tracking in the rest of the transmission blocks of the beacon ($m=1,...,M$).

Channel tracking is necessary as channel estimation requires a higher overhead of pilots and feedback bits, which leads to a shorter duration for data transfer. Thus, at the beginning of transmission blocks $m=1,...,M$, a single pilot with a duration of the one time slot is transmitted to track the estimated channel. Based on the literature, we assume that the time dedicated to channel estimation and tracking is negligible in comparison with that of data transfer. Thus, the channel is considered to be static during this time period \cite{overlap,RACE}.

\begin{figure}[t!]
\centering
\includegraphics[scale=.6]{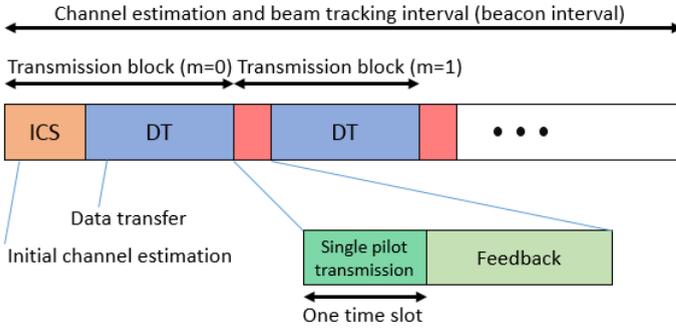}
\caption{Transmission scheme used for communications between the TX and RX.}
\label{TB}
\end{figure}

\subsubsection{Structure of Beamformers}
For the proposed framework, we focus on analog beamforming. However, to further improve the performance, it can be incorporated as part of hybrid beamforming. As proposed by \cite{n11,n12,n13}, in the hybrid structure, beamforming is divided into a digital precoder followed by an analog precoder. The design of the digital precoder for the specifications and codebooks used in our proposed algorithms can be found in \cite{RH}. Therefore, we focus on the analog structure shown in Fig. \ref{TB2}. The proposed framework is explained for a single user for simplicity. Hence, we assume a single RF chain at each node. The mmWave communication system is considered to have $N_\textrm{r}$ antennas at the receiver and $N_\textrm{t}$ antennas at the transmitter (TX) without loss of generality.

\begin{figure}[t!]
\centering
\includegraphics[width=3.5in]{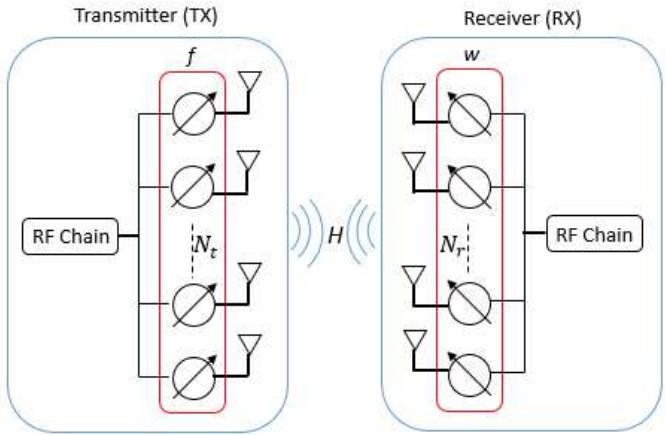}
\caption{Structure of beamformers.}
\label{TB2}
\end{figure}

\subsection{Channel Acquisition Model}
The transmitted pilots are assumed to have unit power and occupy one time slot. If pilot $x$ is transmitted using TX beamformer $\boldsymbol{f}$ ($\|\boldsymbol{f}\|^2 = 1$) and power $P$, the transmitted signal can be expressed as
\begin{equation} \label{n4}
\boldsymbol{s} = \sqrt{P}\boldsymbol{f}x.
\end{equation}
The signal observed by the receiver can be expressed as
\begin{equation} \label{n5}
\boldsymbol{r}_m = \sqrt{P}\boldsymbol{H}_m\boldsymbol{f}x + \boldsymbol{n}_m,
\end{equation}
where the subscript $m$ represents the transmission block number and $\boldsymbol{n}_m$ is an additive white Gaussian noise ($\boldsymbol{n}_m\thicksim \mathcal{C}\mathcal{N}(\boldsymbol{0},N_0\boldsymbol{I})$) imposed on the received signal.
Furthermore, if the combining vector $\boldsymbol{w}$ is applied to received signal $\boldsymbol{r}$, the processed received signal can be written as
\begin{align}
y_m &= \sqrt{P}\boldsymbol{w}^H\boldsymbol{H}_m\boldsymbol{f}x + \boldsymbol{w}^H\boldsymbol{n}_m\nonumber\\
&= \sqrt{P}\boldsymbol{w}^H\boldsymbol{H}_m\boldsymbol{f}x + n_m. \label{n7}
\end{align}
Since $\|\boldsymbol{w}\|^2 = 1$, $n_m$ follows the same distribution as the elements of the vector $\boldsymbol{n}_m$ ($n_m\thicksim \mathcal{C}\mathcal{N}(\boldsymbol{0},N_0)$).

The AoA and AoD of a single path in the $m$-th transmission block are denoted by $\theta_m$ and $\phi_m$, respectively. Assuming uniform linear array (ULA) at both ends of the transmission, the receive and transmit array response vectors are given by
\begin{align}\label{n1}
\setstretch{.9}
\boldsymbol{a}_\textrm{r}(\theta_m) &= \frac{1}{\sqrt{N_\textrm{r}}}[1,e^{-j\frac{2\pi}{\lambda}d\cos\theta_m},...,e^{-j(N_\textrm{r}-1)\frac{2\pi}{\lambda}d\cos\theta_m}]^T\\
\boldsymbol{a}_\textrm{t}(\phi_m) &= \frac{1}{\sqrt{N_\textrm{t}}}[1,e^{-j\frac{2\pi}{\lambda}d\cos\phi_m},...,e^{-j(N_\textrm{t}-1)\frac{2\pi}{\lambda}d\cos\phi_m}]^T,\label{m1}
\end{align}
with $d$ and $\lambda$ denoting the antenna spacing and the carrier wavelength, respectively. For simplicity, we consider a two-dimensional model, and hence, only azimuth angles are considered. In practice, ULA can be steered using phase shifters via a progressive phase shift \cite{r6}.

The channel between the RX and TX is denoted by an $N_\textrm{r} \times N_\textrm{t}$ matrix, i.e., $\boldsymbol{H}_m$. This matrix can be written mathematically as~\cite{RH,rh_track,c2}
\begin{equation} \label{n8}
\boldsymbol{H}_m = \sum\limits_{l=1}^{L}\alpha_m^{(l)}\boldsymbol{a}_\textrm{r}(\theta_{m}^{(l)})\boldsymbol{a}_t(\phi_m^{(l)})^H,
\end{equation}
where the index $l$ implies the $l$-th path and $\alpha_m^{(l)}$ represents the complex path gain of the path. Scattering in mmWave induces more than 20 dB attenuation \cite{r11}. Hence, we consider the LoS component as the target path for the channel estimation and tracking of the vehicle. This is a reasonable assumption as a LoS component has been shown to be the only component, in most cases, that can provide the required reliability for the high transmission rate in mmWave communications \cite{r11,r12}. Moreover, recent measurements have shown mmWave communication channels to be sparse in the geometric domain~\cite{p1}. Hence, paths are less likely to overlap, and we can assume that only one path lies within the main beam direction~\cite{rh_track}, and the other paths fall into side lobes. Therefore, we consider non-LoS paths to be negligible compared to the more dominant LoS component. This assumption becomes more accurate as the number of antennas increases and the beam width grows narrower. Based on this assumption, the observed signal from (\ref{n7}) can be written as
\begin{equation} \label{n9}
\begin{split}
y_m &= \sqrt{P}\alpha_m\boldsymbol{w}^H\boldsymbol{a}_\textrm{r}(\theta_m)\boldsymbol{a}_\textrm{t}(\phi_m)^H\boldsymbol{f}x + n_m.\\
\end{split}
\end{equation}

As our description of the framework focus on one beacon interval, and the channel estimation occurs only once in a beacon at $m=0$, we can simplify the notation by dropping the subscript zero for the channel estimation. Then, the observed pilot can be expressed as
\begin{equation} \label{nn9}
\begin{split}
y &= \sqrt{P}\alpha\boldsymbol{w}^H\boldsymbol{a}_\textrm{r}(\theta)\boldsymbol{a}_\textrm{t}(\phi)^H\boldsymbol{f}x + n.\\
\end{split}
\end{equation}
Starting from $m=1$, transceivers will have the estimated AoA and AoD. Therefore, the pointing direction of the beamformers is adjusted to these angles. Denoting pointing directions of the receiver combiner and the transmit beamformer at the $m$-th transmission block by $\overline{\theta}_m$ and $\overline{\phi}_m$, respectively, the directed beamformers can be expressed as
 \begin{equation} \label{n2}
\boldsymbol{w}(\overline{\theta}_m) = \frac{1}{\sqrt{N_\textrm{r}}}[1,e^{-j\frac{2\pi}{\lambda}d\cos\overline{\theta}_m},...,e^{-j(N_\textrm{r}-1)\frac{2\pi}{\lambda}d\cos\overline{\theta}_m}]^T
\end{equation}
 \begin{equation} \label{n3}
\boldsymbol{f}(\overline{\phi}_m) = \frac{1}{\sqrt{N_\textrm{t}}}[1,e^{-j\frac{2\pi}{\lambda}d\cos\overline{\phi}_m},...,e^{-j(N_\textrm{t}-1)\frac{2\pi}{\lambda}d\cos\overline{\phi}_m}]^T.
\end{equation}
Subsequently, the observed signal can be expressed as

\begin{equation} \label{nn10}
\begin{split}
y_m &= \sqrt{P}\alpha_m\boldsymbol{w}(\overline{\theta}_m)^H\boldsymbol{a}_\textrm{r}(\theta_m)\boldsymbol{a}_\textrm{t}(\phi_m)^H\boldsymbol{f}(\overline{\phi}_m)x + n_m.\\
\end{split}
\end{equation}

\section{Proposed Channel Estimation Algorithm}\label{Initial channel estimation}

In this section, we start by reviewing an extended version of the binary search approach for channel estimation presented in \cite{RH}. Then, we develop a sparse representation of the system, based on maximum likelihood detection (MLD) of the AoA and AoD, and derive a closed-form expression for the optimal number of feedback bits required, subject to PEE constraints. Finally, we explain the proposed RAF algorithm and derive upper and lower bounds on the achievable PEE.

\subsection{Multi-Stage Channel Estimation without PEE Constraints}
Extending the multi-stage approach in\cite{RH}, we divide the whole space for AoAs and AoDs into $K$ subspaces creating $K^2$ possible combinations in each stage. The target path that we wish to estimate is located in one of the possible TX-RX sub-spaces. After estimating the most probable TX-RX subspace pair, they are further divided into another $K$ subspaces, followed by another round of estimation. The process continues until AoD and AoA reach the specified resolution. In the s-th stage, the beamforming vectors at the TX and RX for the $k$-th subspace are represented by $\boldsymbol{f}^s_k$ and $\boldsymbol{w}^s_k$.

To estimate the most likely TX-RX subspaces, a pilot signal ($|x|^2 = 1$) is transmitted in each of the $K^2$ transmitter and receiver angle combinations. Each combination corresponds to one AoA subspace candidate at the receiver and one AoD subspace candidate at the transmitter. As a result, the system can be represented as
\begin{equation} \label{n13}
\boldsymbol{c}^{s,K^2} = \sqrt{P}x\boldsymbol{l}^{s,K^2}+\boldsymbol{n}^{s,K^2},
\end{equation}
where superscripts represent the stage number $s$ and the total number of measurements $K^2$. Also, $\boldsymbol{n}$ is a $K^2\times1$ vector whose elements are i.i.d. Gaussian random variables, and $\boldsymbol{l}^{s,K^2}$ is a vector containing channel responses to all combinations of the transmit and receive beamforming vectors that can be expressed as
\begin{equation} \label{n14}
\setstretch{.9}
\boldsymbol{l}^{s,K^2} =
\begin{bmatrix}
    (\boldsymbol{w}^s_1)^H\boldsymbol{H}\boldsymbol{f}^s_1  \\
    (\boldsymbol{w}^s_2)^H\boldsymbol{H}\boldsymbol{f}^s_1   \\
    \vdots  \\
    (\boldsymbol{w}^s_1)^H\boldsymbol{H}\boldsymbol{f}^s_2  \\
    (\boldsymbol{w}^s_2)^H\boldsymbol{H}\boldsymbol{f}^s_2   \\
    \vdots  \\

    (\boldsymbol{w}^s_K)^H\boldsymbol{H}\boldsymbol{f}^s_{K}
\end{bmatrix}.
\end{equation}
In order to find the desired beamforming vectors, we need the dictionary matrix of all possible steering vectors of angles, shown as
\begin{equation} \label{n15}
\boldsymbol{A}_{\textrm{DIC}} = [ \boldsymbol{a}(0), \boldsymbol{a}(\frac{2\pi}{N_\textrm{t}}),\dots, \boldsymbol{a}(\frac{2\pi(N_\textrm{t}-1)}{N_\textrm{t}})].
\end{equation}
The beamforming vector for the $k$-th subspace at the TX can be found by solving the following equation
\begin{equation} \label{n16}
\boldsymbol{A}_{\textrm{DIC}}^H\boldsymbol{f}_k^{s} = \boldsymbol{z}_i^{s,k},
\end{equation}
where $\boldsymbol{z}_i^{s,k}$ is an $N_\textrm{t}\times 1$ vector, in which values included in the intended transmit subspaces are equal to the constant $C_s$ and zero otherwise. This vector is mathematically defined as:

\begin{equation} \label{alaki}
\boldsymbol{z}_i^{s,k} = \Bigg\{   \begin{tabular}{cc}
  $C_s$ & if $\dfrac{i\pi}{N_t} \in$ $k$th subspace, $i\in\, 1,2,..,N_t-1$ \\
  0, & otherwise
  \end{tabular}
\end{equation}
The value of $C_s$ is chosen to normalize the magnitude of beamforming vectors to unity (i.e., $\|\boldsymbol{f}\|^2 = 1$). From (\ref{n16}), $\boldsymbol{f}_k^{s}$ is calculated as
\begin{equation} \label{n18}
 \boldsymbol{f}_k^{s} = (\boldsymbol{A}_{\textrm{DIC}}\boldsymbol{A}_{\textrm{DIC}}^H)^{-1}\boldsymbol{A}_{\textrm{DIC}}\boldsymbol{z}_i^{s,k}.
\end{equation}
The same procedure is used to find the beamforming vectors of the RX. After $K^2$ measurements, the RX will compare the magnitude of $K^2$ received pilots and choose the one with the largest magnitude. This TX-RX subspace pair would be the most likely pair to include the transmission path.

The channel estimation algorithm above does not offer any PEE guarantees. Therefore, we next extend this algorithm such that the error in each stage is below the specified threshold.
\subsection{A Sparse System Representation}
The explained multi-stage channel estimation algorithm in the previous section conducts $K^2$ measurements in every stage $s$. By substituting $\boldsymbol{H}$ into (\ref{n14}), we therefore have

\begin{equation} \label{n19}
\setstretch{.9}
\boldsymbol{l}^{s,K^2} =
  \alpha
\begin{bmatrix}
    (\boldsymbol{w}^s_1)^H\boldsymbol{a}_r(\theta) \boldsymbol{a}_t(\phi)^H \boldsymbol{f}^s_1  \\
    (\boldsymbol{w}^s_2)^H\boldsymbol{a}_r(\theta) \boldsymbol{a}_t(\phi)^H \boldsymbol{f}^s_1   \\
    \vdots  \\
    (\boldsymbol{w}^s_1)^H\boldsymbol{a}_r(\theta) \boldsymbol{a}_t(\phi)^H \boldsymbol{f}^s_2  \\
    (\boldsymbol{w}^s_2)^H\boldsymbol{a}_r(\theta) \boldsymbol{a}_t(\phi)^H \boldsymbol{f}^s_2   \\
    \vdots  \\
    (\boldsymbol{w}^s_K)^H\boldsymbol{a}_r(\theta) \boldsymbol{a}_t(\phi)^H \boldsymbol{f}^s_{K}
\end{bmatrix}.
\end{equation}
The multiplication of $(\boldsymbol{w}^s_1)^H\boldsymbol{a}_\textrm{r}(\theta)$ and $\boldsymbol{a}_\textrm{t}(\phi)^H \boldsymbol{f}^s_1$ is only non-zero if the AoA and AoD are aligned to the beamforming vectors \cite{RACE,RH,shaham2018raf}. Therefore, only one row of $\boldsymbol{l}^{s,K^2}$ is non-zero \cite{RH}. In this context, we now define a new matrix $\boldsymbol{G}^{s,q}$ that represents its initial state. It can be expressed as
\begin{equation} \label{n20}
 \boldsymbol{G}^{s,q}=  \boldsymbol{G}^{s,K^2}  =  \mathbf{I}_{2K \times 2K},
\end{equation}
where the index $q$ denotes the number of measurements conducted so far.

Finding the AoA and AoD is equivalent to finding a $K^2 \times 1$ vector $\boldsymbol{v}$ that is zero everywhere except the desired row of $\boldsymbol{G}^{s,q}$, in which it is equal to one. Hence, $\boldsymbol{l}^{s,q}$ and the observation vector $\boldsymbol{c}^{s,q}$ can be written as

\begin{equation} \label{n21}
 \boldsymbol{l}^{s,q} = \sqrt{P}xC_s^2 \alpha  \boldsymbol{G}^{s,q}\boldsymbol{v}^T
\end{equation}
\begin{equation} \label{n22}
\boldsymbol{c}^{s,q} = \sqrt{P}xC_s^2 \alpha  \boldsymbol{G}^{s,q}\boldsymbol{v}^T + \boldsymbol{n}^{q}.
\end{equation}
Assuming $d-$th element of $\boldsymbol{v}$ is equal to one, the estimated AoA subspace $\hat{k}_\textrm{t}$ and the AoD subspace $\hat{k}_\textrm{r}$ can be expressed as
\begin{equation} \label{n23}
\hat{k}_\textrm{t} = \lceil \dfrac{d}{K} \rceil ,\; \hat{k}_\textrm{r} = d - K(\hat{k}_\textrm{t} - 1).
\end{equation}
The new presentation of the channel estimation system indicates that the possible outcomes of the channel estimation are equivalent to the rows of matrix $\boldsymbol{G}^{s,q}$.

\subsection{Multi-Stage Channel Estimation with PEE Constraints by Maximum Likelihood Detection}
In our algorithm, the MLD method will be used for the estimation of AoA and AoD. After $q$ measurements, the distribution of the observation vector $\boldsymbol{c}^{s,q}$ can be written as\footnote{The detailed derivation of (\ref{n24}) can be found in \cite{overlap}}
\begin{equation} \label{n24}
    \boldsymbol{c}^{s,q} = \mathcal{C}\mathcal{N}(0, \boldsymbol{\Sigma_v} ),
\end{equation}
where
\begin{align} \label{n25}
\boldsymbol{\Sigma}_{v} &= PC_s^4\boldsymbol{G}^{s,q}\boldsymbol{v}\boldsymbol{v}^T(\boldsymbol{G}^{s,q})^H + N_0 \boldsymbol{I}_{q}.
\end{align}
It can be seen that the received vector follows circularly symmetric complex Gaussian (CSCG) distribution which has the probability density function of
\begin{align} \label{n26}
    f(\boldsymbol{c}^{s,q}&|\boldsymbol{v},\boldsymbol{G}^{s,q}) =\\ & \frac{1}{\pi^{q}\text{det}(\boldsymbol{\Sigma}_{v})} \text{exp}(-(\boldsymbol{c}^{s,q})^H \boldsymbol{\Sigma}^{-1}_{q} {\boldsymbol{c}^{s,q}}). \nonumber
\end{align}
In order to get a better understanding of the probability density, it is useful to inspect them in terms of probability. Defining the set $\mathcal{V}$ as all legitimate $K^2$ outcomes of the vector $\boldsymbol{v}$, the probability can be written as

\begin{align} \label{n27}
   p(\boldsymbol{v}|\boldsymbol{c}^{s,q})  =  \frac{ f(\boldsymbol{c}^{s,q}|\boldsymbol{v}) }{ \sum\limits_{\boldsymbol{j}\in {\mathcal{V}} }  f(\boldsymbol{c}^{s,q}|\boldsymbol{j})   }.
\end{align}
We are looking for the vector $\boldsymbol{j}$ that results in maximum probability by $ \argmax\limits_{\boldsymbol{j} \in {\mathcal{V}}}p(\boldsymbol{v}|\boldsymbol{c}^{s,q})$, which, as explained in the previous section, can be used to find the AoA and AoD. Upon completion of any stage $s$, the estimated channel coefficient can be expressed as

\begin{align}\label{n28}
\hat{\alpha} =  \frac{x (\boldsymbol{w}^s_{\hat{k}_\textrm{r}})^H\boldsymbol{H}\boldsymbol{f}^s_{\hat{k}_\textrm{t}}}{C^2_{s} }.
\end{align}

\begin{algorithm}[t]
\setstretch{.9}
\DontPrintSemicolon 
\textbf{Input:} $N_t$, $N_r$, $K$. \\
// Calculate:\\

$ \{\boldsymbol{f}^s_k\} \  \   $    $\forall k= 1,...,K$\\
$ \{\boldsymbol{w}^s_k\} \  \   $    $\forall k= 1,...,K$

\For {$s<S$} {

\For {$i=1\  to\  K$} {
\For {$j=1\  to\  K$} {
Transmitter transmits using $\boldsymbol{f}^s_i $\\
Receiver measures using $\boldsymbol{w}^s_j$
}
}

// After initial $K^2$ measurements\\
$q\leftarrow K^2$\\
$\boldsymbol{c}^{s,q} \leftarrow \sqrt{P}x\boldsymbol{l}^{s,q}+\boldsymbol{n}^{s,q}$\\
$\boldsymbol{d}\leftarrow\argmax\limits_{\boldsymbol{j} \in {\mathcal{V}}}p(\boldsymbol{v}|\boldsymbol{c}^{s,q})$\\
d$\leftarrow$ non-zero element of $\boldsymbol{d}$\\

$\hat{k}_\textrm{t} \leftarrow \lceil \dfrac{d}{K} \rceil ,\; \hat{k}_\textrm{r} \leftarrow d - K(\hat{k}_\textrm{t} - 1)$\\
\While{$p(\boldsymbol{v}|\boldsymbol{c}^{s,q})< (1 - \Gamma)$
and $q<q_{\textrm{max}}$}{
$q++$\\
Transmitter transmits using $\boldsymbol{f}^S_{\hat{k}_\textrm{t}}$\\
Receiver receives using $\boldsymbol{w}^S_{\hat{k}_\textrm{r}}$\\

// Update:\\
$\boldsymbol{d}\leftarrow\argmax\limits_{\boldsymbol{j} \in {\mathcal{V}}}p(\boldsymbol{v}|\boldsymbol{c}^{s,q})$\\
$d\leftarrow$ non-zero element of $\boldsymbol{d}$\\
$\hat{k}_\textrm{t} \leftarrow \lceil \dfrac{d}{K} \rceil ,\; \hat{k}_\textrm{r} \leftarrow d - K(\hat{k}_\textrm{t} - 1)$\\

}
}

\textbf{Output:}
$\hat{\alpha} =  \frac{x (\boldsymbol{w}^S_{\hat{k}_\textrm{r}})^H\boldsymbol{H}\boldsymbol{f}^S_{\hat{k}_\textrm{t}}}{C^2_{s} }$, $\hat{k}_\textrm{t}$, $\hat{k}_\textrm{r}$.\\
\caption{An algorithm that determines the optimal number of feedback bits.}
\label{SF}
\end{algorithm}

In order to have a benchmark to compare the RAF algorithm’s estimation performance, it is important to know what the optimal number of feedback bits is. This value needs to be large enough to ensure the desired PEE. In other words, we are looking for the minimum implementable number of feedback bits that guarantees the desired PEE. From information-theoretical perspective, the minimum number is one with a single feedback including  $\lceil\log_2(K)\rceil$ bits\cite{ref10}. We verify that this number is achievable by developing an algorithm which only needs $\lceil\log_2(K)\rceil$ bits of feedback. The cost of having the optimal number of feedback bits is a large number of channel measurements. Therefore, this algorithm is just used as a benchmark and can not be a realistic alternative in practice. Finding the AoA and AoD is equivalent to finding a $K^2 \times 1$ vector $\boldsymbol{v}$ that is zero everywhere except the desired row of $\boldsymbol{G}^{s,q}$, in which it is equal to one. Hence, $\boldsymbol{l}^{s,q}$ and the observation vector $\boldsymbol{c}^{s,q}$ can be written as

We denote $\Gamma$ as the probability of the event that a channel estimation is incorrect. The algorithm starts by having the initial $K^2$ measurements which result in the primary channel estimation. The TX continues to send pilots using the same sequence as the initial measurements. After each transmission, using MLD, the RX is capable of calculating $p(\boldsymbol{v}|\boldsymbol{c}^{s,q})$. As soon as reaching the desired PEE by $p(\boldsymbol{v}|\boldsymbol{c}^{s,q})> (1 - \Gamma)$, the RX will feedback $\lceil\log_2(K)\rceil$ bits to notify the TX about the estimated AoD. The process of adding a new measurement for the TX subspace of $\hat{k}_\textrm{t}$ and the RX subspace of $\hat{k}_\textrm{r}$ can be mathematically written as
\begin{align}\label{n29}
\setstretch{.9}
\boldsymbol{c}^{s,q+1} = \sqrt{P} x \left[\begin{array}{ccc} \boldsymbol{l}^{q} \\ (\boldsymbol{w}_{\hat{k}_\textrm{r}}^{s})^H  \boldsymbol{H} \boldsymbol{f}_{\hat{k}_\textrm{t}}^{s} \end{array}
\right] +\left[\begin{array}{ccc} \boldsymbol{n}^{q}  \\ (\boldsymbol{w}_{\hat{k}_\textrm{r}}^{s})^H  n   \end{array}\right],
\end{align}
Note that there is always a probability of outage when the channel power gain is close to zero. In order to prevent an excessive number of measurements, we set a maximum to the number of pilots that could be transmitted, denoted by $q_{\textrm{max}}$. The formal representation of the algorithm is given in Algorithm \ref{SF}.

\subsection{Robust Adaptive Multi-Feedback Algorithm}

Multi-stage channel estimation algorithms are mainly based on a fixed number of channel estimations. As an example, the authors in \cite{RH} used $K^2$ measurements in each stage to estimate the channel. Although the proposed algorithms are effective, they did not consider the performance in terms of the PEE. If due to the additive noise, the detection of the estimated AoA and AoD is incorrect in any of the stages, the algorithms will not be able to estimate the channel correctly. Therefore, devising an algorithm to ensure the desired PEE is crucial. The authors in\cite{RACE}, proposed a rate adaptive algorithm (RACE) in order to reach the desired PEE. Unfortunately, the algorithm requires a high number of channel feedback bits even for $K=2$, particularly at low SNR. Therefore, it is not practical to use the algorithm in fast changing environments, such as V2I scenarios. To this end, we propose an algorithm called RAF. In contrary to the existing algorithms, the RAF algorithm is based on exploiting the estimated channel coefficient to estimate the channel. The significance of using the channel coefficient is the entailed information about the number of measurements required. This helps to estimate the time to commence sending feedback bits and consequently requires a low number of feedback bits as well as pilot transmissions.

Before explaining the algorithm, we use the rudiments of information theory to find a lower bound on the number of measurements. The channel estimation is equivalent to finding a vector  $\boldsymbol{v}$ that contains $K^2$ binary bits encoded into $q$ (number of pilots transmitted) symbols. Therefore, the system has a transmission rate of $ \mathcal{C} = K^2/q$. According to the Shannon-Hartley theorem, we can easily derive the relation infra \cite{ref10}

\begin{align} \label{n30}
\setstretch{.9}
   \mathcal{C}  &= \dfrac{K^2}{q} \leq \log_2(1 + SNR_s)\nonumber \\
   &\rightarrow q \geq \dfrac{K^2}{log_2(1 + SNR_s)},
\end{align}
where $SNR_s$ (in stage s) can be written as
\begin{align} \label{n32}
	SNR_s = \dfrac{|\alpha|^2PK^{(2s-2)}}{N_0}.
\end{align}
Substituting (\ref{n32}) in (\ref{n30}), a lower bound can be found on the number of measurements that is required in each stage conditioned on the estimated value of $\alpha$. After $q$ measurements ($q \geq K^2$), if the mean of observations received in the estimated AoA and AoD is denoted by $\lambda^q$, the value of $\alpha$ can be estimated as

\begin{align}\label{n33}
\hat{\alpha} =  \frac{ \lambda^q}{\sqrt{P} C^2_{s} }.
\end{align}
Therefore, we have a lower bound on the number of measurements required.

\begin{algorithm}
\setstretch{.9}
\DontPrintSemicolon 
\textbf{Input:} $N_t$, $N_r$, $K$. \\

\textbf{Initialization:} .\\

// Calculate:\\
$ \{\boldsymbol{f}^s_k\} \  \   $    $\forall k= 1,...,K$\\
$ \{\boldsymbol{w}^s_k\} \  \   $    $\forall k= 1,...,K$

\For {$s<S$} {

\For {$i=1\  to\  K$} {
\For {$j=1\  to\  K$} {
Transmitter transmits using $\boldsymbol{f}^s_i $\\
Receiver measures using $\boldsymbol{w}^s_j$
}
}

// After initial $K^2$ measurements\\
$q\leftarrow K^2$\\
$\boldsymbol{c}^{s,q} \leftarrow \sqrt{P}x\boldsymbol{l}^{s,q}+\boldsymbol{n}^{s,q}$\\
$\boldsymbol{d}\leftarrow \argmax\limits_{\boldsymbol{j} \in {\mathcal{V}}}p(\boldsymbol{v}|\boldsymbol{c}^{s,q})$\\
$d\leftarrow$ non-zero element of $\boldsymbol{d}$\\
$\hat{k}_\textrm{t} \leftarrow \lceil \dfrac{d}{K} \rceil ,\; \hat{k}_\textrm{r} \leftarrow d - K(\hat{k}_\textrm{t} - 1)$\\

// Find a lower bound for the number of measurements required\\
$\lambda^q\leftarrow$ the mean of values in $\boldsymbol{c}^{s,q}$ corresponding to $\hat{k}_\textrm{t}$ and $\hat{k}_\textrm{r}$\\
$\hat{\alpha} \leftarrow  \frac{ \lambda^q}{\sqrt{P} C^2_{s} }$\\
$L\leftarrow\dfrac{K^2}{log_2(1 + \dfrac{|\hat{\alpha}|^2PK^{(2s-2)}}{N_0})}$

\For {$i=1\  to\  K$} {
\For {$j=1\  to\  K$} {
$q++$\\
Transmitter transmits using $\boldsymbol{f}^s_i $\\
Receiver measures using $\boldsymbol{w}^s_j$\\
Repeat lines $20$ to $22$\\
\If{$q\geq L$}{Break;}
}
\If{$q\geq L$}{Break;}
}

// Update:\\
Repeat lines $16$ to $18$\\

\While{$p(\boldsymbol{v}|\boldsymbol{c}^{s,q})< (1 - \Gamma)$\\
\and $q<q_{\textrm{max}}$}{
$q++$\\
Transmitter transmits using $\boldsymbol{f}^S_{\hat{k}_\textrm{t}}$\\
Receiver receives using $\boldsymbol{w}^S_{\hat{k}_\textrm{r}}$\\

// Update:\\
Repeat lines $16$ to $18$\\

}
}

\textbf{Output:}
$\hat{\alpha} =  \frac{x (\boldsymbol{w}^S_{\hat{k}_\textrm{r}})^H\boldsymbol{H}\boldsymbol{f}^S_{\hat{k}_\textrm{t}}}{C^2_{s} }$, $\hat{k}_\textrm{t}$, $\hat{k}_\textrm{r}$.\\
\caption{Robust adaptive multi-feedback algorithm (RAF).}
\label{Shannon}
\end{algorithm}

In each stage, the RAF algorithm starts by conducting $K^2$ initial channel measurements. The MLD enables the system to have an estimation of the AoA and AoD that can further be used to estimate the value of channel coefficient ($\alpha$). Having the estimated $\alpha$, the receiver can predict a lower bound for the required number of measurements. Up to the point of reaching the PEE threshold, the TX continues to send the pilots as explained in the optimal feedback algorithm. As the pilots are accumulated, the same process of MLD is used to achieve a better estimation of $\alpha$ which results in obtaining a more accurate lower bound. After reaching the PEE threshold, the RX feeds back the estimated AoD. At this point, the TX stops sending the pilots in the order of initial channel estimation and only sends a pilot by the estimated AoD. The RX knows the estimated AoA and utilizes the corresponding combiner to receive the pilot. Following the same process after receiving each pilot, the RX estimates the AoA and AoD and feeds back the estimated AoD. The stage terminates as soon as the required estimation precision is reached. In the final transmission of feedback bits, an extra bit will be transmitted to notify the transmitter to stop the transmission of pilot signals. The RAF algorithm is represented formally in Algorithm \ref{Shannon}.

\subsection{RAF Performance Analysis}
To analyze the performance of RAF, in this section we proceed to derive upper and lower bounds on its achievable PEE. Before doing so, we formally define PEE in Definition~\ref{def7}.

\begin{defn}\label{def7}
	\textit{Probability of estimation error (PEE):} At any stage $s$, assuming all detections in previous stages have been correct, we define PEE as
    \begin{equation} \label{ff1}
    p(EE|\boldsymbol{G}^{s,q},\boldsymbol{v})= p(\boldsymbol{v}\neq \boldsymbol{\hat{v}}),
    \end{equation}
    where by $EE$ we refer to estimation error and $p(\boldsymbol{v}\neq \boldsymbol{\hat{v}})$ indicates the probability of an event in which the estimated vector $\boldsymbol{\hat{v}}$ is not equal to the transmitted vector $\boldsymbol{v}$.
\end{defn}

From an information theoretic perspective, we are encoding a vector $v$ entailing $K^2$ bits over $q$ measurements using a generator matrix $\boldsymbol{G}^{s,q}$. On the receiving side, we are observing the signal $\boldsymbol{c}^{s,q}$ from which we estimate the transmitted symbol (i.e., vector $v$). The PEE can be written in terms of the union of possible outcomes as
\begin{equation} \label{f2}
p(\boldsymbol{v}\neq \boldsymbol{\hat{v}}) = \bigcup_{\boldsymbol{\hat{v}} \in \mathcal{V} ,\boldsymbol{v}\neq \boldsymbol{\hat{v}}} p(\boldsymbol{c}^{s,q}\rightarrow \boldsymbol{\hat{c}}^{s,q})
\end{equation}
where $\boldsymbol{\hat{c}}^{s,q}$ is the observation vector corresponding to $\boldsymbol{\hat{v}}$ and the term $p(\boldsymbol{c}^{s,q}\rightarrow \boldsymbol{\hat{c}}^{s,q})$ indicates the probability of an event in which $\boldsymbol{\hat{c}}^{s,q}$ is chosen as the outcome over $\boldsymbol{c}^{s,q}$. As we are using MLD to detect the received symbol, referencing to \cite{s1}, the pairwise probability of error estimation over the fading channel with channel coefficient of $\alpha \thicksim \mathcal{N}(0,Q)$ can be calculated as
\begin{equation} \label{f3}
p(\boldsymbol{c}^{s,q}\rightarrow \boldsymbol{\hat{c}}^{s,q}) = 0.5-\sqrt{\dfrac{\Omega^2}{8+4\Omega^2}},
\end{equation}
where $\Omega$ is given by
\begin{equation} \label{f4}
\Omega = \sqrt{\dfrac{PQC_s^4}{2N_0}}\|\boldsymbol{G}^{s,q}(\boldsymbol{v}- \boldsymbol{\hat{v}})  \|^2.
\end{equation}
Having the pairwise probability of error estimation and conditioning on the transmitted vector $\boldsymbol{v}$, PEE for the given generator matrix can be written as
\begin{align} \label{f5}
\setstretch{.9}
p(EE|\boldsymbol{G}^{s,q})&= \sum_{\boldsymbol{v}\in \mathcal{V}} p(\boldsymbol{v})p(EE|\boldsymbol{G}^{s,q},\boldsymbol{v})\nonumber \\
&=\sum_{\boldsymbol{v}\in \mathcal{V}} p(\boldsymbol{v})\bigcup_{\boldsymbol{\hat{v}} \in \mathcal{V} ,\boldsymbol{v}\neq \boldsymbol{\hat{v}}} p(\boldsymbol{c}^{s,q}\rightarrow \boldsymbol{\hat{c}}^{s,q})\nonumber \\
&=\sum_{\boldsymbol{v}\in \mathcal{V}} p(\boldsymbol{v})\bigcup_{\boldsymbol{\hat{v}} \in \mathcal{V} ,\boldsymbol{v}\neq \boldsymbol{\hat{v}}}\bigg(\frac{1}{2}-\sqrt{\dfrac{\Omega^2}{8+4\Omega^2}}\  \bigg).
\end{align}
An upper bound on (\ref{f5}) can be applied by replacing union with summation, which results in

\begin{align} \label{f8}
p(EE|\boldsymbol{G}^{s,q})
\leq \sum_{\boldsymbol{v}\in \mathcal{V}}\sum_{\substack{\boldsymbol{\hat{v}} \in \mathcal{V}\\ \boldsymbol{v}\neq \boldsymbol{\hat{v}}}} p(\boldsymbol{v})\bigg(\frac{1}{2}-\sqrt{\dfrac{\Omega^2}{8+4\Omega^2}}\  \bigg).
\end{align}
Based on (\ref{f8}), an upper bound for the PEE over all stages $1 \textrm{ to }S$ can be calculated by

\begin{align} \label{f9}
p(EE) &= 1- \prod_{s=1}^{S}(1-p(EE|\boldsymbol{G}^{s,q}))\nonumber \\
&\leq  \sum_{s=1}^{S}\sum_{\boldsymbol{v}\in \mathcal{V}}\sum_{\substack{\boldsymbol{\hat{v}} \in \mathcal{V}\\ \boldsymbol{v}\neq \boldsymbol{\hat{v}}}} p(\boldsymbol{v})\bigg(.5-\sqrt{\dfrac{\Omega^2}{8+4\Omega^2}}\  \bigg).
\end{align}
Also, assuming that all possible realizations of the set $\mathcal{V}$ have an equal probability, a lower bound on PEE can be derived as
\begin{align} \label{f10}
p(EE|\boldsymbol{G}^{s,q})
\geq \frac{1}{2}-\sqrt{\dfrac{PQC_s^4\|(\boldsymbol{v}- \boldsymbol{\hat{v}})  \|^2}{16N_0+PQC_s^4\|(\boldsymbol{v}- \boldsymbol{\hat{v}})  \|^2}}.
\end{align}

\section{Proposed Beam Tracking Algorithm}\label{Proposed beam tracking algorithm}
In this section, the proposed system model for the EKF algorithm is explained. First, the necessity for a new state evolution model is demonstrated, followed by the proposed model. Then, the observation model corresponding to the state evolution model is derived, and finally, the EKF algorithm is illustrated. The advantages of our proposed mmWave channel tracking model for vehicular communications can be summarized as:

\begin{itemize}
  \item In contrast to the previous models, which are based on the angles of arrival and departure, our proposed model is based on position, velocity and channel coefficient, resulting in a more practical approach for mmWave vehicular communications by increasing the tracking performance and lowering the number of time channel estimation is required.
  \item Based on the new model, we are able to consider factors such as the velocity of the vehicle, block duration, etc. As addressed in this section, if angles are used as state variables for the movement of vehicles, the state evolution model is no longer a linear one. Such non-linearity leads to high complexity in the calculation of Jacobians, which cannot be easily implemented in practical mmWave vehicular networks. Hence, our proposed model can avoid such problems.
  \item Our proposed linear model results in a much lower complexity in the calculation of Jacobian matrices, as it is will be shown by the derivation of closed-form expressions for Jacobians.\\
\end{itemize}

\subsection{State Evolution Model}
Previous attempts to apply the EKF algorithm on mmWave beam tracking were based on using the AoA and AoD as state variables \cite{zhang,rh_track}. The existing state evolution models are linear and assumed to evolve by a Gaussian noise with zero mean. Also, the hypothetical noise is additive that highly facilitates the processing of realistic scenarios. Unfortunately, such modeling is not realistic for most of the vehicular communication scenarios. This drawback becomes evident when considering how actually the angles evolve as shown in Fig \ref{f1}.

\begin{thm}
If a vehicle moves from the transmit angle $\phi_{j}$ to $\phi_{j+1}$, the change in the transmit angle and similarly for the receiving angle can be calculated by
\begin{align} \label{n34}
\phi_{j+1}-\phi_{j} = -\cot^{-1}(\frac{h}{\cos^2\phi_{j}(v_j +w_j)\Delta t} - \tan\phi_{j}),
\end{align}
where $v_j$ is the velocity of the vehicle at the first location and $w_j$ is a Gaussian noise.
\end{thm}

\begin{proof}
Considering Fig. \ref{f1} and denoting length of the line $AC$ by $T$, for simplicity, we apply the law of sines to $\overset{\triangle}{ABC}$ and $\overset{\triangle}{ACD}$ to yield the relations below:

\begin{equation}\label{n35}
\dfrac{\Delta d}{\sin (\Delta \phi)} = \dfrac{T}{\sin(90- \phi _j)}
\end{equation}
and
\begin{equation}\label{n36}
\dfrac{T}{\sin (90)} = \dfrac{h}{\sin(90- \phi _{j+1})}.
\end{equation}
Substituting $T$ from (\ref{n36}) into (\ref{n35}), $\Delta{d}$ is derived as
\begin{equation}\label{n37}
\Delta{d}=\dfrac{\sin (\Delta \phi)h}{\operatorname{cos}(\phi _j)\operatorname{cos}(\phi _{j+1})}.
\end{equation}
Furthermore, $\operatorname{cos}(\phi _{j+1})$ can be written as
\begin{multline}\label{n38}
\operatorname{cos}(\phi _{j+1}) = \cos (\phi _{j+1} - \phi_j+\phi _j)\\ =\cos (\Delta \phi)\operatorname{cos}(\phi _j)+\sin(\Delta \phi) \sin(\phi _j)
\end{multline}
and by substituting it into (\ref{n37}), the Kinematic formula relating velocity to displacement is derived to be
\begin{multline}\label{n39}
\Delta d= (v_j +w_j)\Delta t\\ = \dfrac{\sin (\Delta\phi)h}{\operatorname{cos}(\phi _j)(\cos (\Delta \phi)\operatorname{cos}(\phi _j)+\sin(\Delta \phi) \sin(\phi _j))}.
\end{multline}
Solving (\ref{n39}) for $\Delta \phi$ directly results in (\ref{n34}).
\end{proof}

\begin{figure*}[!t] 
\normalsize
\begin{multline}
\setstretch{.9}
\frac{\partial y_m}{\partial d_m} = \sum\limits_{q=0}^{N_\textrm{t}-1} \sum\limits_{p=0}^{N_\textrm{r}-1} \dfrac{\frac{\sqrt{P}\alpha_mx}{N_\textrm{t}N_\textrm{r}}j\frac{2\pi}{\lambda}dh^2(p+q)}{\sqrt{(h^2+(d_{m-1} + v_{m-1}\Delta t)^2)^3}} \times \\ e^{\dfrac{j\frac{2\pi}{\lambda}d(b_{pq}\sqrt{h^2+(d_{m-1} + v_{m-1}\Delta t)^2} + (p+q)(d_{m-1} + v_{m-1}\Delta t)  )  }{\sqrt{h^2+(d_{m-1} + v_{m-1}\Delta t)^2}}} \label{n46}
\end{multline}
\hrulefill
\end{figure*}

Recall that the main drawback of the EKF algorithm is the complexity in its implementation. It can be seen that (\ref{n34}) is non-linear with respect to the angle. Also, the hypothetical noise that happens due to the change in velocity is non-additive. These are the two main factors affecting the complexity. Therefore, if we use the angles as state variables, the calculation of the Jacobians for EKF algorithm is of high complexity. To solve this problem, we propose to use the position, velocity, and complex channel coefficient as the state variables, which give a linear state model and hypothetical additive noise for vehicular communications. As a result, the state vector can be written as

\begin{equation} \label{n40}
\boldsymbol{x}_m = [d_m,v_m,\alpha_m^\textrm{R},\alpha_m^\textrm{I}]^T,
\end{equation}
where $d_m$ and $v_m$ denote the position and velocity of the vehicle at the $m$-th transmission block, respectively. The Gaussian coefficient is divided into a real part and an imaginary part, i.e., $\alpha_m = \alpha_m^\textrm{R} + j\alpha_m^\textrm{I}$, which helps to have the state vector as real numbers.  $\alpha_m^\textrm{R}$ and $\alpha_m^\textrm{I}$ are assumed to follow the first order Gauss-Markov model expressed by \cite{rh_track}
\begin{align} \label{n41}
\setstretch{.9}
\alpha_{m+1}^\textrm{R} &=\rho \alpha_m^\textrm{R} + \xi_m\\
\alpha_{m+1}^\textrm{I} &=\rho \alpha_m^\textrm{I} + \xi'_m,
\end{align}
where $\rho$ is the correlation coefficient, $\xi_m$, $\xi'_m \thicksim \mathcal{N}(0,\dfrac{1-\rho^2}{2})$, and  $\xi[-1]$, $\xi'_[-1] \thicksim \mathcal{N}(0,\dfrac{1}{2})$. The evolution of position and velocity are thus formulated as
\begin{align} \label{n42}
\setstretch{.9}
d_{m+1} &= d_m + v_m\Delta t + w_m\Delta t\\
v_{m+1} &= v_m + w_m,
\end{align}
with $w_m$ denoting hypothetical noise that represents the change in the speed of the vehicle. It is assumed to follow the Gaussian distribution, i.e., $w_m\thicksim \mathcal{N}(0,\sigma_w^2)$. In summary, the state evolution equation can be written as
\begin{equation} \label{n43}
\boldsymbol{x}_{m+1} =\boldsymbol{A} \boldsymbol{x}_m + \boldsymbol{u}_m,
\end{equation}

where
\begin{align} \label{n44}
\setstretch{.9}
\boldsymbol{A} =
\begin{bmatrix}
1\; \Delta t\;  0\;\;  0\\
0 \;\; 1 \;\; 0\;\; 0\\
0\;\;  0 \;\; 1\;\;  0\\
0\;\;  0 \;\; 0 \;\; 1
\end{bmatrix},
\end{align}

and $\boldsymbol{u}_m\thicksim \mathcal{N}(0,\boldsymbol{\Sigma_u})$ with

\begin{align} \label{n45}
\boldsymbol{\Sigma_u} = \mathrm{diag}(\,[\,(\Delta t\sigma_w)^2, (\sigma_w)^2, 1 - \rho^2, 1 - \rho^2\,]\,).
\end{align}

\subsection{Observation Expression}
In order to complete the model for the EKF algorithm, we need to derive the measurement function in terms of the state variables. 
By substituting the (\ref{n1},\ref{m1},\ref{n2},\ref{n3}) in (\ref{nn10}), we have

 \begin{equation} \label{n47}
\begin{split}
y_m &= \frac{\sqrt{P}\alpha_mx}{N_\textrm{t}N_\textrm{r}}(\sum\limits_{p=0}^{N_\textrm{r}-1}e^{-j\frac{2\pi}{\lambda}dp(\cos\theta_m-\cos\overline{\theta}_m)})\times\\&\,\,\,\,\,\,\,(\sum\limits_{q=0}^{N_\textrm{t}-1}e^{j\frac{2\pi}{\lambda}dq(\cos\phi_m-\cos\overline{\phi}_m)})+ n_m\\
&=\frac{\sqrt{P}\alpha_m}{N_\textrm{t}N_\textrm{r}} \sum\limits_{q=0}^{N_\textrm{t}-1} \sum\limits_{p=0}^{N_\textrm{r}-1}e^{j\frac{2\pi}{\lambda}d(-p\cos\theta_m+q\cos\phi_m+b_{pq})}+ n_m,
\end{split}
\end{equation}

 where
\begin{equation} \label{n48}
b_{pq} = p\cos\overline{\theta}_m- q\cos\overline{\phi}_m.
\end{equation}
In Fig. \ref{f1}, the AoA and AoD of the system can be measured as

\begin{equation} \label{n49}
\theta_m = \operatorname{atan2}(\dfrac{-h}{-d_m}) = \operatorname{atan2}(\dfrac{-h}{-(d_{m-1} + v_{m-1}\Delta t)})\\
\end{equation}
\begin{equation} \label{n50}
\phi_m = \operatorname{atan2}(\dfrac{h}{d_m}) = \operatorname{atan2}(\dfrac{h}{d_{m-1} + v_{m-1}\Delta t}).
\end{equation}
where $\operatorname{atan2}$ is the four-quadrant inverse tangent. The cosine of the angles are calculated as
\footnote{$\operatorname{cos}(\operatorname{atan2}(\dfrac{y}{x})) = \dfrac{y}{\sqrt{x^2+y^2}}$.}
\begin{align}\label{n51}
\operatorname{cos}(\theta_m)&= \dfrac{-(d_{m-1} + v_{m-1}\Delta t)}{\sqrt{h^2+(d_{m-1} + v_{m-1}\Delta t)^2}}\\
\operatorname{cos}(\phi_m)&= \dfrac{(d_{m-1} + v_{m-1}\Delta t)}{\sqrt{h^2+(d_{m-1} + v_{m-1}\Delta t)^2}} \label{n52}
\end{align}
Substituting equations (\ref{n51}), (\ref{n52}) into (\ref{n47}), we have the observation equation in terms of the state variables as

 \begin{equation} \label{n53}
\begin{split}
y_m &=\frac{\sqrt{P}\alpha_mx}{N_\textrm{t}N_\textrm{r}}\times \\ &\sum\limits_{q=0}^{N_\textrm{t}-1} \sum\limits_{p=0}^{N_\textrm{r}-1}e^{j\frac{2\pi}{\lambda}d(\dfrac{(p+q)(d_{m-1} + v_{m-1}\Delta t)}{\sqrt{h^2+(d_{m-1} + v_{m-1}\Delta t)^2}}+b_{pq})}\\ &+ n_m = g(\boldsymbol{x}_m)+n_m.
\end{split}
\end{equation}

\subsection{EKF-Based Beam Tracking}
In this subsection, we will present how to use the EKF algorithm to track the vehicle. As the vehicle moves, the state vector of the process evolves. Our aim is to match the $\overline{\theta}_m$ and $\overline{\phi}_m$ to $\theta_m$ and $\phi_m$, respectively.

EKF recursion procedure \cite{ekf} is described in the Fig. \ref{recursion}. For estimating the state at the ($m+1$)-th transmission block, the algorithm starts by assigning the predicted values of ($m+1$)-th state estimate and its covariance to the values of $m$-th transmission block. Then, the Kalman gain is calculated based on the assigned values. Finally, obtaining a new observation and using the values calculated for Kalman gain, the ($m+1$)-th state and its covariance are updated.

\begin{figure}[t]
\centering
\includegraphics[scale=.35]{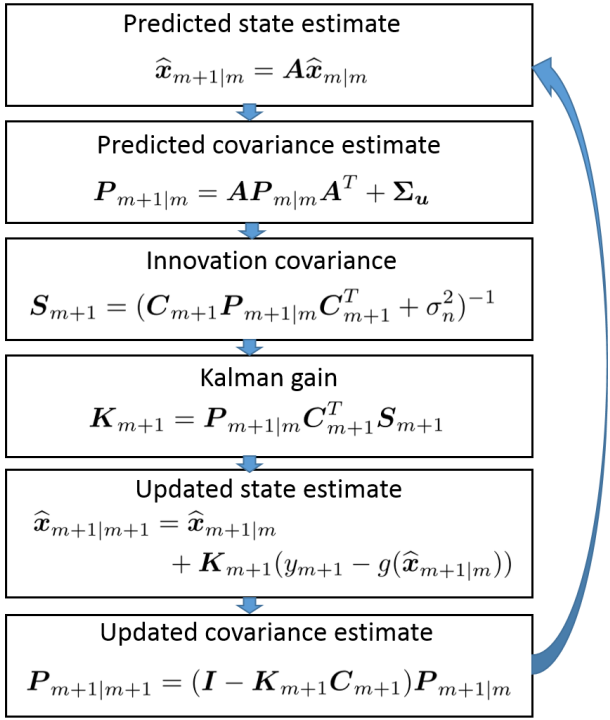}
\centering
\caption{Flowchart of the EKF recursion.}
\label{recursion}
\end{figure}
In Fig. \ref{recursion}, $\boldsymbol{C}_{m+1}$ is the observation transition matrix defined by the following Jacobian matrix
\begin{equation} \label{n60}
\boldsymbol{C}_{m+1} = \left.\frac{\partial g}{\partial\boldsymbol{x}}\right|_{\widehat{\boldsymbol{x}}_{m+1|m}}.
\end{equation}
The partial derivative with respect to position is given in (\ref{n46}) and the partial derivative with respect to velocity is calculated by
\begin{equation}\label{n61}
\frac{\partial y_m}{\partial v_m} = \frac{\partial y_m}{\partial d_m} \times \Delta t.
\end{equation}
For the channel coefficient, calculation of the partial derivative is straightforward, which can be done by simply excluding the noise term and channel coefficient in (\ref{n46}).
Note that, in order to deal with real numbers in implementation of the EKF algorithm, $y_m$ and $\boldsymbol{C}_m$ are substituted by $\widetilde{y}_m = [y_m^\textrm{R}, y_m^\textrm{I}]^T$ and $\widetilde{\boldsymbol{C}}_m = [\boldsymbol{C}_m^\textrm{R}, \boldsymbol{C}_m^\textrm{I}]^T$ in the equations used in the Fig. \ref{recursion}.

\section{Numerical Results }\label{aaa}
\subsection{Simulation Configurations}

In simulations, the number of antennas at both the TX and RX is set to be $64$, with $\lambda/2$ spacing, the channel coefficient is assumed to follow a Gaussian distribution with zero mean and unit covariance, i.e., $\mathcal{C}\mathcal{N}(0,1)$, and the initial AoD and AoA are set to $-135$ and $45$ degrees. Furthermore, we set the value of $\rho$ to $0.995$, $\Delta t$ to  $0.001$ s, the initial speed of vehicle to $60$ km/h, and the variation in speed $\sigma_w $ to $1.4 $ m/s. Also, in order to compare our proposed channel estimation approach to previously established methods in \cite{RH} and \cite{RACE}, we set the value of $K$ to two, the maximum number of measurements $q_{\textrm{max}}$ to $264$, and the target PEE to $ 10^{-2}$.

\subsection{Achievable PEEs}
Achieving the target PEE is essential for guaranteeing a reliable transmission between the TX and RX. Fig. \ref{pee} represents the PEE for different values of SNR. We have compared our results with the current state of art algorithm termed RACE, proposed in \cite{RACE}, and the fixed-rate algorithm proposed in \cite{RH}. The fixed-rate algorithm only considers channel estimation, without any constraints on PEE, and expectedly, the algorithm is not able to maintain the PEE below the predefined threshold of $10^{-2}$. As can be seen in the figure, the RACE and RAF algorithms both result in a comparable PEE performance achieving the desired PEE with a negligible difference. Note that the PEE performance of both RAF and RACE algorithms are over the PEE threshold in low SNRs. The reason behind this behavior is the existing probability of outage in the simulated system.

To verify our analysis, we present the analytical and numerical results in Fig. \ref{Analytical}. As can be seen in the figure, numerical results lie within the analytical lower and upper bounds, which verifies the credibility of our simulations. Note that the analytical lower bound is not very tight compared with the numerical results due to the assumption in the calculation of the lower bound that $p(\boldsymbol{v})$ is uniformly distributed.

\begin{figure}[t!]
\centering
\includegraphics[width=3.548in]{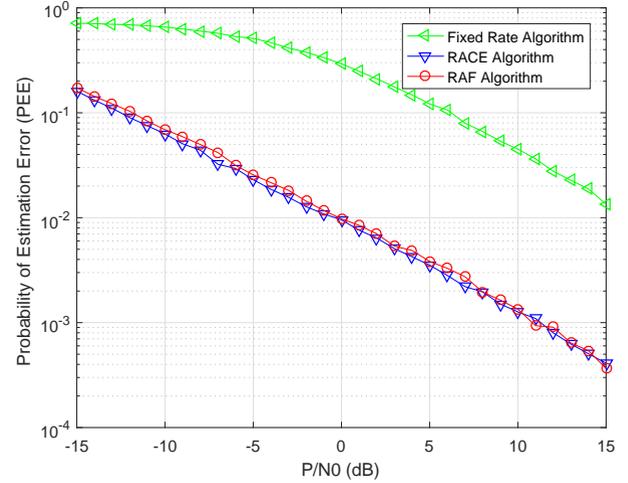}
\caption{PEE performance of the RAF algorithm in comparison with the algorithms in \cite{RACE} and \cite{RH}.}
\label{pee}
\end{figure}

\begin{figure}[t!]
\centering
\includegraphics[width=3.548in]{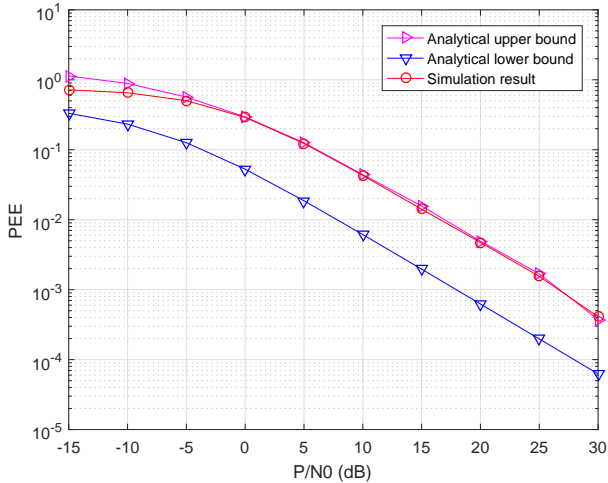}
\caption{Comparison of the numerical and analytical results for PEE.}
\label{Analytical}
\end{figure}

\subsection{Channel Estimation Time}\label{Overall duration of channel estimation}
The primary metric for measuring the effectiveness of channel estimation algorithms is the overall time required to estimate the channel. In this subsection, we consider the overall time required for the channel estimation and compare the RAF algorithm with the previous approaches. Note that having a lower channel estimation period is essential to maximize the time assigned for data transmission.

Transmitting a pilot or a feedback bit both require only one time slot to be conducted. Therefore, the overall time needed for the channel estimation can be calculated by the addition of time slots assigned for pilot and feedback bit transmissions. We study these two factors in the following.

Recall that algorithms are required to ensure a predetermined PEE in addition to estimating the channel. For this reason, we do not consider the fixed-rate algorithm, as it provides no assurance for the required PEE. Fig. \ref{feedback} exhibits the overhead feedback performance of the RAF algorithm compared with the RACE algorithm. The optimal feedback bits required is also shown in the figure. Achieving a lower number of feedback bits is desirable as it results in more time for data transmission. It can be seen that the RAF algorithm requires a significantly small number of feedback bits compared with the RACE algorithm. The average feedback bits needed is almost as low as the optimal number. The difference between the algorithms becomes apparent, particularly in the low SNR regime. Therefore, the RAF algorithm is considered to be a viable option for dedicating time for data transmission while ensuring that error probability is in an acceptable range.

The overall time required for channel estimation based on the RAF and RACE algorithms is shown in Fig. \ref{s}. This figure illustrates the superior performance of the RAF algorithm. On average, for SNRs ranging from -15 dB to 15 dB, the performance is improved by $14\%$. At low SNR, the difference is more significant. For instance, in the SNR of -15dB, the overall time required for the channel estimation is reduced by $30\%$ using the RAF algorithm. Therefore, the RAF algorithm is able to increase the time assigned for the data transmission significantly while maintaining the PEE below the predefined threshold.

\begin{figure}[t!]
\centering
\includegraphics[width=3.548in]{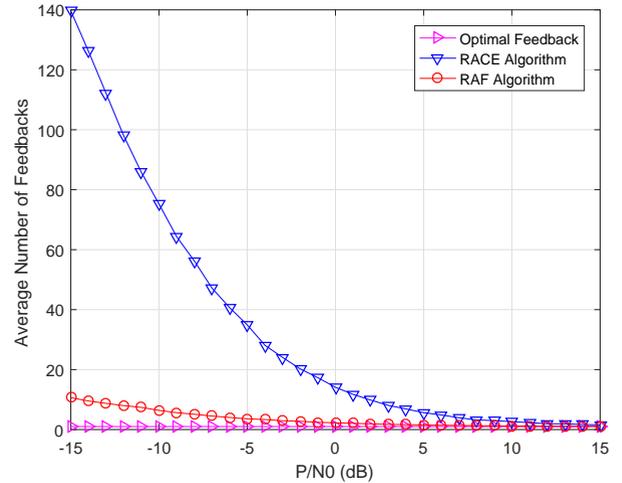}
\caption{Feedback performance of the RAF algorithm compared to \cite{RACE} and the optimal number of feedback bits.}
\label{feedback}
\end{figure}

\begin{figure}[t!]
\centering
\includegraphics[width=3.548in]{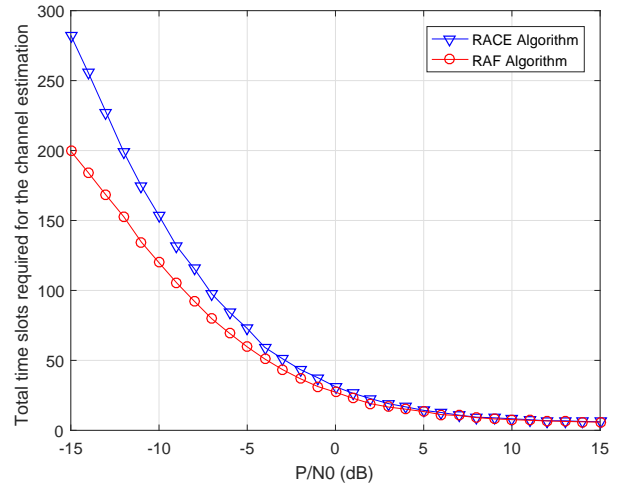}
\caption{Comparison of the time required for the channel estimation in each stage.}
\label{s}
\end{figure}

\subsection{Performance of the Proposed Beam Tracking Algorithm}

The main factors affecting the tracking performance are the received SNR and the number of antennas used at the TX and RX. The mean square error (MSE) performance of the EKF tracking algorithm for various SNR settings over $3000$ experiments is shown in Fig. \ref{s2}. The performance is shown for the AoD, and similarly, it can also be shown for the AoA as they are both functions of the position. The valid tracking threshold is chosen to be $\sqrt{E[|\phi_m-\overline{\phi}_m|^2 ]} = BW/2$ \cite{rh_track}, where $BW$ denotes the half-power beamwidth of the antenna array. Therefore, we define that the tracking is lost if the MSE is larger than the threshold. Such a threshold is indicated by a horizontal line in this figure. Note that, half-power beamwidth is a function of the beam direction, it is in its maximum for the end-fire direction and minimum for the broadside direction\cite{antenna}. Logically, we choose the broadside direction as our threshold so that our results stay valid in generic scenarios. The broadside direction of the beam happens when the vehicle is exactly below the antenna array, and its value can be estimated as $\lambda /(dN)$.

\subsubsection{SNR Performance}

Fig. \ref{s2} presents the performance of our proposed tracking model for the SNRs of -5, 0, and 5 dB. As can in the figure, the tracking performance crosses the threshold in 24, 31, and 38 transmission blocks for the SNRs of minus five, zero and five, respectively. Note that each transmission block corresponds to 1 ms, and the tracking time can be calculated accordingly. The most imminent trend that can be observed is that the valid duration of beam tracking becomes larger as the SNR value increases. Therefore, expectedly, we are able to track the vehicle for a more extended period of time.

The most relevant prior work to our scheme is the one proposed in \cite{rh_track}. The authors used AoA, AoD, and the channel coefficient as state variables. It was assumed that the angles evolve using a hypothetical noise with zero mean and variance of $(\dfrac{.5}{180}\pi)^2$, which is unable to characterize how a vehicle moves in the real environment. Moreover, to be consistent with the results in \cite{rh_track}, the number of antennas is set to 16. The results of the comparison are shown in Fig. \ref{s4}.

As can be seen in Fig. \ref{s4}, our proposed model increases the valid tracking duration from 85 transmission blocks to 105 transmission blocks. The reason for the improvement in tracking performance is that our proposed model considers the dynamics of the system, such as the velocity of the vehicle, block duration, etc. Moreover, as addressed in this paper, if angles are used as state variables for the movement of the vehicle, the state evolution model is no longer linear. Such non-linearity leads to high complexity in the calculation of Jacobians, which cannot be easily implemented in practical mmWave vehicular networks.

\subsubsection{Effect of Number of Antennas Used}

Another critical factor that was not considered in prior works is the number of antennas used at the transceivers. This number was set to 16, irrespective of consequences that larger antenna arrays may cause. In practice, mmWave communication systems are more likely to have large antenna arrays to compensate for the path loss that is severe with shorter wavelengths. Table \ref{t2} lists the valid tracking durations of the previous work compared to our proposed model for various numbers of antennas. It can be seen that increasing the number of antennas significantly shortens the valid duration of the beam tracking. For instance, once equipping the transceivers with 64 antennas, the previous model can only track the user on average for four successive transmission blocks, whereas our proposed model can extend the tracking duration on average to $31$ transmission blocks. To sum up, considering the dynamics of vehicular communications, our proposed model is able to considerably improve the tracking period in comparison with the existing methods.

\begin{table}[h!]
\caption{Comparison of the valid tracking duration (in terms of transmission blocks) for different number of antennas at SNR $=0 $ dB.}
\centering
\begin{tabular}{|>{\centering\arraybackslash}m{1.8cm} || >{\centering\arraybackslash}m{1.1cm} | >{\centering\arraybackslash}m{1.1cm} | >{\centering\arraybackslash}m{1.1cm} | >{\centering\arraybackslash}m{1.1cm} |}
 \hline
  $N$ &  16 & 32 & 64 & 128\\
 \hline
 \hline Prior model in \cite{rh_track} & 85 &  19 & 4 &  3  \\
 \hline Our proposed model &  105 &  62 &  31 &  8\\
 \hline Improvement &  20 & 43 &  27 &  15\\

 \hline
\end{tabular}
\label{t2}
\end{table}

\begin{figure}[t!]
\centering
\includegraphics[width=3.548in,height=2.67in]{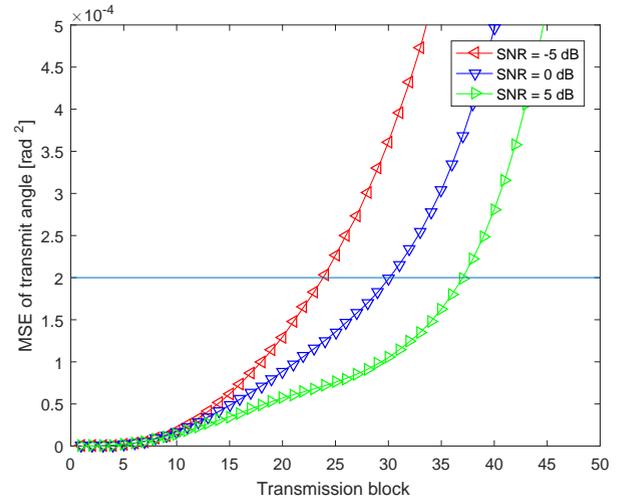}
\caption{Effect of SNR on tracking performance of the EKF algorithm.}
\label{s2}
\end{figure}

\begin{figure}[t!]
\centering
\includegraphics[width=3.548in,height=2.67in]{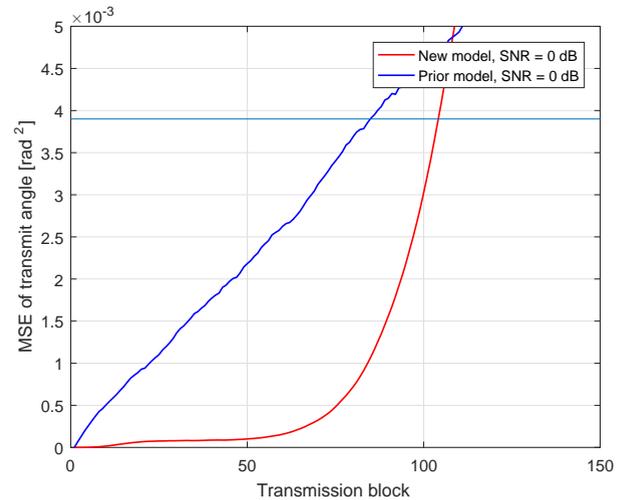}
\caption{Comparison of the proposed model for vehicular communications to the prior work presented in \cite{rh_track}.}
\label{s4}
\end{figure}

\subsection{Effect of Channel Estimation on Beam Tracking}

After considering the performance of the RAF channel estimation algorithm and the proposed model for beam tracking, we investigate the impact of the channel estimation error on the beam tracking performance of the vehicle. The results in Table \ref{t3} demonstrates how having an inaccurate channel estimation can lead to a decrease in the valid duration of beam tracking. For this experiment, the number of antennas is set to $16$, and the channel estimation errors are chosen based on the half-power, calculated as $BW = \dfrac{\lambda}{dN}$. The maximum possible error ($BW/2$) is divided into four equal divisions, and the result of 3000 runs of each is illustrated. It is evident from the table that a larger estimation error directly deteriorates the performance of the beam tracking. Therefore, having a robust algorithm such as RAF is necessary to ensure that the probability of estimation error is below an acceptable threshold according to the system requirements.

\begin{table}[th]
\caption{Impact of channel estimation error on beam tracking performance at SNR $=0$ dB (in terms of transmission blocks).}
\centering
\begin{tabular}{|>{\centering\arraybackslash}m{2.2cm} || >{\centering\arraybackslash}m{.7cm} | >{\centering\arraybackslash}m{.7cm} | >{\centering\arraybackslash}m{.7cm} | >{\centering\arraybackslash}m{.7cm} |>{\centering\arraybackslash}m{.7cm}|}
 \hline
  Channel estimation error& 0& $\dfrac{BW}{8}$  & $\dfrac{BW}{4}$  & $\dfrac{3BW}{8}$ &$\dfrac{BW}{2}$  \\
 \hline Valid tracking duration & 105 &  93 & 85 &  72&64  \\

 \hline
\end{tabular}
\label{t3}
\end{table}

\section{Conclusion }\label{conclusion}
In this paper, we proposed a novel channel estimation and beam tracking algorithm suitable for mmWave communications in V2I communications. The proposed channel estimation algorithm, dubbed RAF, was shown to be capable of significantly reducing the required channel estimation time by $14\%$ and the feedback overhead by $75.5\%$ on average, in comparison with the existing algorithms. We have also re-investigated the implementation of the EKF beam tracking for mmWave vehicular communications. Accordingly, new state evolution and observation models have been proposed considering the vehicle's position and velocity as well as channel coefficient. These models are shown to be more practical for vehicular to infrastructure communications.
\bibliographystyle{IEEEtran}
\bibliography{Journal_CS_T}

\EOD
\end{document}